\documentclass[journal]{IEEEtran}
\usepackage{cite}
\usepackage[pdftex]{graphicx}
\usepackage{amsmath,amssymb,amsfonts,amsthm}
\usepackage[caption=false,font=normalsize,labelfont=sf,textfont=sf]{subfig}
\usepackage{stfloats}
\usepackage{enumitem}

\newtheorem{theorem}{Proposition}
\newtheorem{question}{Question}
\newtheorem{lemma}{Lemma}
\newtheorem{defn}{Definition}
\newtheorem{asspt}{Assumption}

\begin{document}
\title{Non-Bayesian Social Learning with Imperfect Private Signal Structure}

\author{
	\uppercase{Sannyuya Liu}, 
	\uppercase{Zhonghua Yan*(\lowercase{hit\_yan@mail.ccnu.edu.cn})},

	\uppercase{Xiufeng Cheng} and
	\uppercase{Liang Zhao}
}

\maketitle

\begin{abstract}

	As one of the classic models that describe the belief dynamics over social networks, non-Bayesian social learning model assumes that members in the network possess accurate signal knowledge through the process of Bayesian inference.
	In order to make the non-Bayesian social learning model more applicable to human and animal societies, this paper extended this model by assuming the existence of private signal structure bias: each social member in each time step uses an imperfect signal knowledge to form its Bayesian part belief, then incorporates its neighbors' beliefs into this Bayesian part belief to form a new belief report.
	First, we investigated the intrinsic learning ability of an isolated agent and deduced the conditions that the signal structure needs to satisfy for this isolated agent to make an eventually correct decision.
	According to these conditions, agents' signal structures were further divided into three different types, "conservative," "radical," and "negative". Then we switched the context from isolated agents to a connected network, our propositions and simulations show the conservative agents are the dominant force for the social network to learn the real state, while the other two types might prevent the network from successful learning. 
	Although fragilities do exist in non-Bayesian social learning mechanism, "be more conservative" and "avoid overconfidence" could be effective strategies for each agent in the real social networks to collectively improve social learning processes and results.

\end{abstract}

\maketitle

\section{Introduction}
\label{sec:introduction}

\IEEEPARstart{L}{earning}
the truth of the world from signals we observed and information provided by others is an essential part of our daily social life.
For most of the time, we hope that our decisions, beliefs, and opinions about the objective world could become more and more precise with the information provided by our social neighbors, especially when our knowledge is inaccurate or inadequate to deal with what we observed. What is the mechanism by which individuals use social connections to form beliefs and make decisions? It is one of the key questions that various social learning models try to answer\cite{lewis2011network, easley2010networks, golub2017learning, chamley2013models, conradt2003group}.

The non-Bayesian social learning model proposed in \cite{jadbabaie2012non} is a model that describes the belief dynamics over the social network.
In this model, each agent of the finite network receives discrete time signal inputs from the objective world and generate discrete time belief outputs. Specifically, in each discrete time step, two information sources are used to update an agent's belief.
One information source is the individual information source, which includes agent's belief report at the last time step and its newly observed signal,  the agent uses this information source to form its personal belief through Bayesian inference.
Another information source is the social information source, which includes the last time step belief reports of all its neighbors.
At the end of the new time step, each agent's new belief report is a convex combination of the personal belief and neighbors' beliefs.
The non-Bayesian social learning model applies well to real social networks, because it takes factors from social, psychological, physical and cognitive aspects into account, these factors include (1) forgetting mechanism, (2) communication, (3) individual's naive style to aggregate information, (4) individual's Bayesian style to make decisions, and (5) individuals are incapable of engaging in fully rational learning.
Along with the model, the authors in \cite{jadbabaie2012non} give sufficient conditions for a network to eventually and collectively learn the world's actual state. 
Their results show that naive information aggregation rule can achieve successful learning, even agents in the network are unaware of important aspects like the network topology and signal structure of other agents.

One default assumption of the model of \cite{jadbabaie2012non} is that, in order to form the individual part belief in each time step, each agent must perform Bayesian inference based on signal statistical models about all possible states.
Each agent's signal spaces and signal statistical models about possible states may vary from each other, but it assumes that everyone in the network possesses the accurate signal statistical model about the real state.

However, in many social contexts, this ideal and harsh condition cannot be met: While network members are unaware of the actual state of the world, they may also use inaccurate knowledge to make Bayesian inference.
There are many reasons for this, such as lacking precise knowledge, personality, attitudes, preferences, and so on.
Therefore, it is necessary to extend the non-Bayesian social learning model by considering this imperfect signal knowledge assumption.
Based on the non-ideal assumption that each agent possesses an imperfect private signal structure, the mathematical contributions of our work are listed as follows:

1. For an isolated agent, we show that different signal structure conditions ensure different types of asymptotic learning (Proposition \ref{THM:INDIV_RADICAL} convergence in probability and Proposition \ref{THM:INDIV_CONSERV} almost sure convergence). Accordingly, agents' signal structures are further divided into three different types, conservative, radical and negative.

2. In the network context, we show that individual almost sure convergence could lead to collective almost sure convergence (Proposition \ref{THM:NETWORKALMOSTSURE}).
Counter-intuitively, we show individual convergence in probability does not guarantee collective convergence in any sense  (Proposition \ref{THM:NETWORKUNCERTAIN}). 

3. We show that in the imperfect signal structure scenario, there is an approximately linear relationship between learning speed and property of signal structure (Proposition \ref{THM:LEARNINGRATE}).

Based on the above mathematical and simulation results, in Section \ref{sec:discuss}, we find some mathematical fragilities in non-Bayesian social learning mechanism.
We further explain the mathematically fragile yet practically effective social learning mechanism by introducing some realistic assumptions into the ideal model, such as to give self-reliance a meaning of self-confidence and allow agents to adjust their signal structures with the change of social pressure.
With these assumptions, we propose signal structure adjusting strategies by which the individuals could collectively improve the learning of the network.
At last, some typical decision-making patterns are extracted based on these adjusting strategies. These decision-making patterns match not only private signal structure types but also people's decision-making styles, which also reveals the implementation value of our work.

Our propositions, numerical simulations, analyses, and conclusions could inspire the in-depth understanding of how different types of members influence the evolution of beliefs in real life social networks.

\section{Related Literature}
Social learning in microeconomics studies is sequential, or in an expanding topology \cite{banerjee1992simple, golub2010naive, acemoglu2011bayesian}. 
In \cite{golub2010naive}, the authors interpret the "wisdom of crowd" in an expanding network mathematically, showing the group intelligence is a network version of the law of large numbers.
In \cite{acemoglu2011bayesian}, the authors discussed the learning problem in a sequential social learning model, where the newly joined node could make rational decisions based on a signal and some predecessors' decisions. 

The non-Bayesian social learning model we adopt in this paper is a type of DeGroot-style learning \cite{degroot1974reaching}, which allows agents in the network to merge neighbors' opinions in a weighted averaging approach in each time step. 
Compared with sequential learning, DeGroot model based on Markov chain theory provides a system perspective to understand social learning. 
In empirical research \cite{chandrasekhar2012testing}, it is verified that DeGroot model could well describe the opinion dynamics in a social network. 
DeGroot model is the foundation of many following social learning models \cite{friedkin1990social, gale2003bayesian}. 
In \cite{friedkin1990social}, the authors model the opinion dynamics of the social network from the angle of the control system with defined inputs and discrete time opinion outputs. 
Model in \cite{gale2003bayesian} introduces Bayesian inference into the network; each agent in the network receives a signal at the beginning, then repeatedly, it observes its neighbors' actions and makes fully rational decisions.
In subsequent research \cite{choi2005learning}, the authors conducted a series of experiments to validate this Bayesian learning model within a three-person network.
One common feature of above studies and non-Bayesian social learning model \cite{jadbabaie2012non} is that they are all social physics model researches targeting human and animal social networks, which means the nodes in networks are assumed to be social members, and edges in the network are assumed to be social connections.

In this paper, we adopt the non-Bayesian social learning model in \cite{jadbabaie2012non}, which uses static topology, naive information aggregation rule, Bayesian self-renewal and allows constant arrival of new signals. 
Compared with sequential learning with expanding topologies, the non-Bayesian social learning model is more like a dynamic system with discrete and distributed signal inputs and belief outputs. 
Our work is closely related to the researches on extensions or variations of non-Bayesian social learning model, such as \cite{molavi2018theory, qipeng2015distributed, nedic2017fast, jadbabaie2013information, liu2011non}.
Some of them focus on variations of learning rule, such as in \cite{molavi2018theory} the authors bridge the gap between different information aggregation rules, and traits of different learning rules are discussed.
In \cite{qipeng2015distributed}, the authors conducted a comprehensive comparison of different learning rules.
In \cite{nedic2017fast}, the authors consider a scenario where each agent's signal structure about the underlying state may be not the actual distribution of signal, which is very similar to our imperfect signal structure scenario; a difference between their work and ours is that our work adopts a naive information aggregation rule that applies to human and animal social networks, while their work applies better to computer decision-making systems (we discuss the effects of different learning rules in Section \ref{sec:social_learning}).
In \cite{jadbabaie2013information}, the authors define the speed of learning, give it an analytical form of bounds, uncover the relationship between signal structure, network topology, and network's learning speed; in our work, we investigate the learning speed problem under our assumptions. 

In addition to these model extensions and variations, studies on the characteristics and performance of non-Bayesian social learning are also relevent to our work, such as \cite{liu2013social, jiang2015indian, rosas2017technological, wang2015social, ho2015robust}. 
In \cite{jiang2015indian}, the performance (learning speed) of non-Bayesian social learning is discussed in a gaming model. 
In \cite{liu2013social} the authors introduce a confidence radius into non-Bayesian social learning model and demonstrate a clustering phenomenon in their model, their analysis approach for the roles played by different types of agents is similar to our work. 
In \cite{rosas2017technological}, information cascades phenomena in social learning models (including the non-Bayesian social learning model) are discussed in detail.

Our work is also inspired by social science studies which reveal inefficiencies or fragilities of social learning, such as \cite{noth2002information, lorenz2011social, pentland2013beyond}. In \cite{noth2002information} information cascades and overconfidence are presented and discussed in a strictly designed experiment.
In \cite{lorenz2011social}, the authors use a series of experiments to demonstrate and analyze statistical, social and psychological effects that ruin the wisdom of crowd.
In \cite{pentland2013beyond}, the author reveals the echo chamber effect of social learning and offers advises about how to avoid it. 
Besides, in \cite{mengel2011decision} the authors discuss the participants' decision-making features in imperfect signal structure scenario via well-designed experiments, which is similar to the scene in our work. 
Compared with these quantitative and qualitative studies, our work tends to give mathematical explanations for observational learning and collective decision-making in social networks and to reveal insights of social learning mechanism that are difficult to reveal by empirical studies.

\section{Model, Assumptions and Main Question}

We consider a finite set $V=\{1,2,\dots,N\}$ containing $N$ agents interacting over a social network, which could be represented by a directed graph $G=(V,E)$. The element $(i,j)$ in the edge set $E$ represents a directed edge connecting agent $i$ to agent $j$, which captures the fact that agent $j$ has access to the belief held by agent $i$.
For each agent $i$, define $\mathcal{N}_i=\{j\in V:(j,i)\in E\}$, called agent $i$'s neighbor set.
Elements in set $\mathcal{N}_i$ are agents whose beliefs are available to agent $i$.
Agent $i$ is isolated if $\mathcal{N}_i=\Phi$, which means this agent has no neighbors.
A directed path in $G=(V,E)$ from agent $i$ to agent $j$ is a sequence of agents starting with $i$ and ending with $j$ such that each agent is a neighbor of the next agent in the sequence.
The social network would be strongly connected if there is at least one directed path from each agent to any other agents.

We write $M_{i\cdot}$ and $M_{\cdot j}$ to denote the $i$-th row and $j$-th column of a matrix $M$.
$M$ is stochastic if its entries are non-negative and the sum of entries in each row is equal to 1.
We define $|\Theta|$ as the cardinality of a set $\Theta$, $\Delta\Theta$ as the distribution over set $\Theta$. Transpose of vector $v$ is denoted by $v^T$.
When we take the logarithm of a vector, we mean taking the logarithm of the vector entry-wise.

\subsection{World Signal Structure}

Let $\theta$ denote a possible state of the objective world, representing an underlying reason for a social or natural phenomenon; all possible states constitute a finite state set $\Theta$. The real state is the actual fact, denoted by $\theta_r,\theta_r\in\Theta$, which in this paper is predefined and static over time. Though the actual state of the world remains unobservable to the individuals, agents in the network keep making repeated noisy observations about $\theta_r$ over discrete time.
For each agent $i$, at the beginning of each time period $t\geq1$, a signal $\omega_{i,t}\in S_i$ drawn according to distribution $g_i\in \Delta S_i$ is observed by agent $i$, here $S_i$ denotes agent $i$'s signal space.
In this paper, all impossible signals are excluded from $S_i$, which means all entries in $g_i$ are strictly positive.
For the reason that signal distribution $g_i$ is the world signal's statistical distribution based on the underlying state $\theta_r$, we name $g_i$ the world signal structure. Besides, we assume that both $g_i$ and $S_i$ may vary across agents but keep static over time.

\subsection{Belief and Private Signal Structure}\label{sec:bapss}
An agent's belief is the extent to which it believes a state is real. An agent's beliefs on all possible states constitute a belief profile.
Belief profile is dynamic over time and varies across agents; thus we could use column vector $\mu_{i,t}$ to denote agent $i$'s belief profile.
Since belief profile is a distribution over state set $\Theta$, we have $\mu_{i,t}\in\Delta\Theta$. In addition, we use $\mu_{i,t}(\theta_m)$ to denote agent $i$'s belief on a specific state $\theta_m$ at time $t$.

Agent $i$'s private signal structure is its statistical knowledge about the signals with respect to all possible states. In the non-Bayesian social learning model, the private signal structure is used to make Bayesian inference. The private signal structure can be represented by a $|S_{i}|\times |\Theta|$ signal-state likelihood mapping 
$$L_{i}=\left[\ell_i^1,\ell_i^2,\cdots,\ell_i^m,\cdots \right].$$
The $m$-th column entry $\ell_i^m$ of $L_{i}$ is the possibility all signals attributed to $\theta_m$ from the view of agent $i$.
For private signal structure, it must satisfy that $\ell_i^m\in\Delta S_{i}$ for all $m$, which also indicates that the transpose of the matrix $L_{i}$ is stochastic.

Since signal space $S_{i}$ is a local property and state space $\Theta$ is defined globally, private signal structures of different agents may have different numbers of rows, but share the same number of columns.

When $\ell_i^r = g_i^T$, which is an ideal assumption indicating that each agent possesses the accurate signal knowledge, it is the scenario that has been discussed in \cite{jadbabaie2012non}. Our work extends this ideal assumption and allows for the private signal structure bias.
\begin{asspt}
	\label{asspt:imperfect}
	Each agent's knowledge about the statistical model of the real state may be imperfect, which means for each agent's private signal structure, it may satisfy that $\ell_i^r \ne g_i^T.$
\end{asspt}

In real social networks, noises, which ruins the perfectness of private signal structure, could be easily introduced into private signal structure by experience, personality, mental status and other unknown endogenous and exogenous factors. Thus Assumption \ref{asspt:imperfect} can significantly extend the application scenario of non-Bayesian social learning model, especially in describing the belief dynamics of human and animal social networks.

\subsection{Agent's Information Sources}
One key issue in agent $i$'s self-renewal process is about its information sources, including what it knows and what it doesn't when a new time period $t+1$ arrives: 
(a) As a basic assumption of our research, agent $i$ has no idea of the world signal structure. 
(b) Since observations are independent, agent $i$ has no access to other agents' even his neighbors' signals.
(c) Agent $i$ forgets the history signals $\omega_{i,\tau}(\tau\le t)$ it observed and most of its history belief profiles $\mu_{i,\tau}(\tau<t)$, which is referred to the \textit{imperfect recall} principle in \cite{molavi2018theory}.
The information sources for the agent to update its belief profile in period $t+1$ include: (1) $\mu_{i,t}$, its belief profile at time $t$, (2) $\omega_{i,t+1}$, the signal observed at the beginning of the new time period, and (3) $\mu_{j,t}(j\in\mathcal{N}_i)$, its neighbors' belief profiles at time $t$.
Each agent in each new time step uses these information sources together with its private signal structure to form its new belief report.

\subsection{Learning Rules and Successful Learning}
The updating of all agent's beliefs over the network is a dynamic process.
Learning rule describes how an agent uses its information sources to update its belief in a new time step.
In \cite{qipeng2015distributed} it has been discussed that the order of Bayesian updating and information aggregation makes little difference. Therefore our model makes Bayesian updating first and takes a weighted averaging information aggregation rule in accordance with models in \cite{jadbabaie2012non} and \cite{jadbabaie2013information}. The weighted averaging information aggregation rule is a DeGroot-style rule satisfying \textit{monotonicity}, \textit{label neutrality}, and \textit{separability} mentioned in \cite{molavi2018theory}, it is a naive and the simplest way to aggregate information from others.

In the non-Bayesian social learning model, agent $i$'s belief updating formula can be written as
\begin{equation}\label{eq:learning_rule1}
	\begin{split}
		\mu_{i,t+1}(\theta_m)&=a_{ii}\mu_{i,t}(\theta_m)\frac{{\ell_i^m}(\omega_{i,t+1})}{d_{i,t}(\omega_{i,t+1})} \\
		& + \sum_{j\in \mathcal{N}_i}{a_{ij}{\mu_{j,t}}(\theta_m)}
	\end{split}
\end{equation}
for all $\theta_m\in\Theta$. For agent $i$, the first part on the right side of Eq.\eqref{eq:learning_rule1} is the Bayesian part of the belief update, which is a result of Bayesian inference based on the posterior belief in period $t$, signal it observed in period $t+1$, and private signal structure. The denominator
\begin{equation*}
	d_{i,t}(\omega_{i,t+1})=\sum_{m=1}^M\ell_i^m(\omega_{i,t+1})\mu_{i,t}(\theta_m)
\end{equation*}
is the one-step forecast, where $a_{ii}$ is the weight referred to the self-reliance. Additionally, the second part on the right side of Eq.\eqref{eq:learning_rule1} is a weighted average of the beliefs held by its neighbors, known as the social part.
$a_{ij}$ captures the weight that agent $i$ assigns to neighbor $j$, and it must satisfy that $a_{ii}+\sum_{j\in\mathcal{N}_i}a_{ij}=1$.

If we use $A=[a_{ij}]$ to denote the network's influence matrix and use $\mu_{t}(\theta_m)$ to denote a vector containing all agents' beliefs on state $\theta_m$ at time $t$, then the belief dynamics of the underlying network can be given by
\begin{equation}\label{eq:learning_rule2}
	\begin{split}
		\mu_{t+1}(\theta_m)=&A\mu_t(\theta_m) + diag\bigg( a_{11}\left[ \frac{\ell_1^m(\omega_{1,t+1})}{d_{1,t}(\omega_{1,t+1})}-1 \right], \\
		&a_{22}\left[ \frac{\ell_2^m(\omega_{2,t+1})}{d_{2,t}(\omega_{2,t+1})}-1 \right], \cdots, \\
		&a_{NN}\left[ \frac{\ell_N^m(\omega_{N,t+1})}{{d_{N,t}}(\omega_{N,t+1})}-1 \right] \bigg)\mu_t(\theta_m)
	\end{split}
\end{equation}
for all $\theta_m\in\Theta$.

Specifically, if we assign the same self-reliance to all agents, Eq.\eqref{eq:learning_rule1} can be rewritten as
\begin{equation}\label{eq:learning_rule3}
	\begin{split}
		{\mu_{i,t+1}}(\theta_m)&=\gamma {\mu_{i,t}}(\theta_m)\frac{\ell_i^m(\omega_{i,t+1})}{{d_{i,t}}(\omega_{i,t+1})} \\
		& + \sum_{j\in \mathcal{N}_i}a_{ij}{\mu_{j,t}}(\theta_m)
	\end{split}
\end{equation} for all $\theta_m\in\Theta$, where $\gamma$ is the common self-reliance. We note that Eq.\eqref{eq:learning_rule3} will be used to investigate learning speed in Sub-section \ref{sec:learning_speed}.

Ulteriorly, if each agent assigns the same weight on its neighbors' beliefs, Eq.\eqref{eq:learning_rule3} can be rewritten as
\begin{equation}\label{eq:learning_rule4}
	\begin{split}
		{\mu_{i,t+1}}(\theta_m)&=\gamma {\mu_{i,t}}(\theta_m)\frac{\ell_i^m(\omega_{i,t+1})}{{d_{i,t}}(\omega_{i,t+1})} \\
		& + \frac{1-\gamma }{\left| \mathcal{N}_i \right|}\sum_{j\in \mathcal{N}_i}{{\mu_{j,t}}(\theta_m)}
	\end{split}
\end{equation} for all $\theta_m\in\Theta$. We note that Eq.\eqref{eq:learning_rule4} is used in all belief dynamics simulations in this paper for simplicity.

We write $\mathbf{1}_x$ to denote a $|\Theta|$-dimension vector with all zero entries except for its $x$-th entry, which is equal to 1.
If the belief profile converges to $\mathbf{1}_r$ for every agent $i$ in the network, we say that the network could achieve successful learning, which means all agents in the network could collectively and eventually learn the real state of the objective world.

\subsection{The Binary Case}

Now let us consider the simplest case: Agent $i$ possesses a $2\times2$ private signal structure matrix, which means both of state set $\Theta$ and signal space $S_i$ are binary. In this work, this simplest case is called \textit{binary case} for short. It should be noted that all numerical simulations in this study are conducted in the binary case, but all our propositions could generalize to $|\Theta|\ge2$ and $|S_i|\ge2$ cases.

In the binary case, let the state set $\Theta = \{\theta_1,\theta_2\}$ and all agents share the same signal set $S_{i}=\{s_1,s_2\}$. We specify $s_1$ is the high signal and $s_2$ is the low signal \cite{easley2010networks}, which means $g_i(s_1)>1/2$ and $g_i(s_2)<1/2$. Let $\theta_1$ be the predefined real state, $\ell_i^1(s_1)=\alpha$ and $\ell_i^2(s_1)=\beta$ for all $i$, then the private signal structure in the binary case can be defined by
$$
 	L_i=\left[ \begin{matrix}
			\alpha   & \beta   \\
			1-\alpha & 1-\beta \\
		\end{matrix} \right].
$$

To some extent, the binary case is quite representative. First, the binary case could cover most of the decision-making scenarios, for state space and signal space of most decision-making problems are binary. Second, as discussed in \cite{jadbabaie2013information}, when the successful learning occurs, the learning process, in the long run, is only related to the wrong state whose signal structure has the minimum KL divergence to the real state, which is referred to \textit{the most indistinguishable alternative state}. Besides, the mathematical merit of the binary signal space is that the signal structure only has 2 degrees of freedom; thus we could demonstrate the property of a signal structure in the unit square $\alpha,\beta\in[0,1]$, which is referred to $\alpha\beta$ square for short. In contrast, when $|\Theta|>2$ or $|S_i|>2$, visualization of signal structure properties is far more complicated.

\subsection{Main Question}
\label{sec:mainResults}

In the case when each agent's private signal structure is perfect, which means $\ell_i^r=g_i^T$ for all $i$, the authors in \cite{jadbabaie2012non} give the conditions enabling successful learning.
Since our work is an extension of \cite{jadbabaie2012non}, firstly, we need some of these successful learning conditions to be our assumptions.

\begin{asspt}
	\label{asspt:positiveSelfReliance}
	All agents have strictly positive self-reliance, which means $a_{ii}>0$ for all $i$.
\end{asspt}

Assumption \ref{asspt:positiveSelfReliance} ensures all agents' signal structures act on social learning. If $a_{ii}=0$ for all ${i}$, the social network is a dynamic system without signal input and Bayesian inference, it will turn into a belief averaging network with naive learning rule.

\begin{asspt}
	\label{asspt:onePositiveInitialBelief}
	There exists at least one agent $i$ with a positive initial belief on the real state $\theta_r$, which means there exists at least one agent $i$ that satisfies $\mu_{i,0}(\theta_r)>0$.
\end{asspt}

Assumption \ref{asspt:onePositiveInitialBelief} is the necessary condition about agents' initial beliefs; without it, there will be no chance for agents in the network to assign a positive belief on real state via Bayesian inference.

With all the above three assumptions, the main question we are going to answer is as follows:
\begin{question}\label{ques:social_learning}
	What are the private signal structure conditions ensuring the successful learning over the network?
\end{question}

In the following section, we show that $\ell_i^r$, which is an agent's knowledge about the real state, plays a major role in individual decision-making.

\section{Individual learning}\label{sec:inividual_learning}

First, we need to look at an agent's intrinsic learning ability by disregarding the social influence. Therefore the following question should be answered first:
\begin{question}\label{ques:individual_learning}
	Without the information from its neighbors, what conditions could ensure an agent's successful learning?
\end{question}

By answering this question, the private signal structure properties that related to the long run result of repeated Bayesian belief renewal will be explored.

\subsection{Positive and Negative Signal Structures}

For an isolated agent, due to lack of information from the network, when a new time period $t+1$ begins, its external information source is only the newly arrived signal $\omega_{i,t+1}$. By removing the social part from the right-hand side of Eq.\eqref{eq:learning_rule1} and assigning $a_{ii}=1$, the belief updating rule of an isolated agent can be simplified as

\begin{equation}
	\label{eq:isolatedLearningRule}
	{\mu_{i,t+1}}(\theta_m)={\mu_{i,t}}(\theta_m)\frac{\ell_i^m(\omega_{i,t+1})}{{d_{i,t}}(\omega_{i,t+1})}
\end{equation}
for all $\theta_m\in\Theta$.

Before digging into the isolated agent's repeated self-renewal process described by Eq.\eqref{eq:isolatedLearningRule}, we introduce the definition of relative entropy (KL divergence) about two distributions.

\begin{defn}
	Given two discrete probability distributions $p$, $q$ with identical supports, the \textbf{relative entropy} of $q$ with respect to $p$ is
	$$D(p||q)=\sum_{j}p_j\log\frac{p_j}{q_j}.$$
\end{defn}

With the above definition, for any state $\theta_m\in\Theta$, let
$$
	h_i(g_i, \theta_m)= D(g_i||\ell_i^m),
$$
then $h_i(g_i, \theta_m)$ could be used to measure the agent's $i$'s \textit{expected information content} per observation in favor of the hypothesis that the underlying state is $\theta_m$. 
In addition, let
\begin{equation}\label{eq:higmr}
	\begin{split}
		h_i^g(\theta_m,\theta_r) &= h_i(g_i, \theta_r) - h_i(g_i, \theta_m) \\
		&= \sum_{s\in S_i}g_i(s)\log\frac{\ell_i^m(s)}{\ell_i^r(s)},
	\end{split}
\end{equation}
then $h_i^g(\theta_m,\theta_r)$ is a measure of the \textit{expected information content} per observation, in favor of (the hypothesis that the underlying state is) ${\theta_r}$ other than ${\theta_m}$. Notice that $h_i^g(\theta_m,\theta_r) = - h_i^g(\theta_r,\theta_m)$.

\begin{figure}
	\centering
	\includegraphics[width=0.35\textwidth]{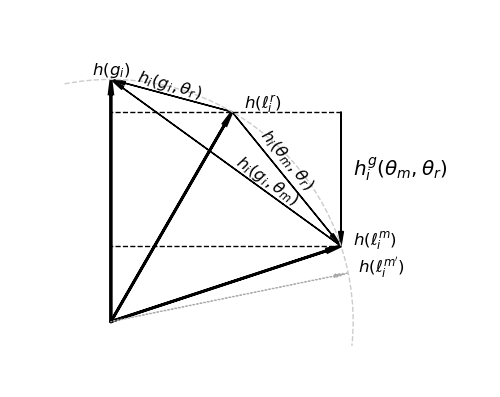}
	\caption{$h_i^g(\theta_m,\theta_r)$ could be understood as the "projection" of $h_i(\theta_m, \theta_r)$ onto the entropy of world signal structure $h(g_i)$.}
	\label{fig:projection}
\end{figure}

In Fig.\ref{fig:projection}, entropies and relative entropies are vectorized $h_i^g(\theta_m, \theta_r)$ into the unit circle\footnote{Generally speaking, vectorization of entropies and relative entropies into the unit circle is of no meaning; Fig.\ref{fig:projection} provides a way to understand the meaning of $h_i^g(\theta_m,\theta_r)$ from the vector space perspective.}.
Fig.\ref{fig:projection} illustrates the relationship between $h_i^g(\theta_m,\theta_r)$ and $h_i(\theta_m, \theta_r)$: $h_i^g(\theta_m, \theta_r)$ could be understood as the "projection" of $h_i(\theta_m, \theta_r)$ onto the entropy of world signal structure $h(g_i)$. If $h_i^g(\theta_m, \theta_r)$ has the same direction with $h(g_i)$, $h_i^g(\theta_m, \theta_r)$ is positive; otherwise, $h_i^g(\theta_m, \theta_r)$ is negative.

With the definition of $h_i^g(\theta_m, \theta_r)$, the condition that an isolated agent should satisfy to achieve successful learning by learning rule Eq.\eqref{eq:isolatedLearningRule} is presented in Proposition \ref{THM:INDIV_RADICAL}.

\begin{theorem}
	\label{THM:INDIV_RADICAL}
	Suppose that agent $i$ has a positive initial belief on state $\theta_r$, and its private signal structure $L_{i}$ satisfies
	\begin{equation}\label{eq:inProbabilityCondition}
		\begin{split}
			h_i^g(\theta_m,\theta_r) < 0
		\end{split}
	\end{equation}
	for all $m\ne r$, then, agent $i$ learns the state $\theta_r$, \textbf{in probability}, by the belief updating rule given by Eq.\eqref{eq:isolatedLearningRule}, which means $\mu_{i,t}\xrightarrow[t\to\infty]{p}\mathbf{1}_r$.

	Proof: \textnormal{See appendix.}
\end{theorem}

Proposition \ref{THM:INDIV_RADICAL} reveals that, by repeated signal observation and application of the Bayesian inference, an agent can learn the underlying state if the underlying state's corresponding column $\ell_i^r$ (in its private signal structure) has the minimum KL divergence from the world signal structure. In other words, the isolated agent will learn the state whose corresponding column best describes the observed signals. A similar case is also mentioned in Section 5.3.1 in \cite{bernardo2009bayesian}.

The bold words "in probability" in Proposition \ref{THM:INDIV_RADICAL} indicates the type of convergence. In particular, Eq.\eqref{eq:inProbabilityCondition} ensures a weak convergence. Another point to note is that if $h_i^g(\theta_m,\theta_r)>0$ for some $m\ne r$, there must exist an $r^\prime\ne r$ where $h_i^g(\theta_{m^\prime},\theta_{r^\prime})<0$ for all $m^\prime\ne r^\prime$. That is to say, an agent must learn a state, no matter whether the state is the underlying state or not.

Then we have the definitions of positive and negative signal structure:
\begin{defn}
	Agent $i$ and its private signal structure $L_i$ are
	\begin{enumerate}[label=\alph*)]
		\item \textbf{positive} if $h_i^g(\theta_m,\theta_r)<0$ for all $m\ne r.$
		\item \textbf{negative} if there exists $\hat m(\hat m\ne r)$ such that $h_i^g(\theta_{\hat m},\theta_r)>0$.
	\end{enumerate}
	\label{def:2}
\end{defn}

For a specific agent, in its signal structure, if the column concerning the real state has the minimum KL divergence to the world signal structure, it is positive. Otherwise, if the column with the minimum KL divergence to the world signal structure is not the $r$-th column, it is negative. For individual learning, positive agents could learn the real state independently, while isolated negative agents always learn a wrong state.

In \cite{jadbabaie2012non}, the authors mention an observational equivalent concept: If there exists a state $\theta_m(m\ne r)$ with $\ell_m=\ell_r$, then $\theta_r$ and $\theta_m$ are observational equivalent for agent $i$. 
Under Assumption \ref{asspt:imperfect}, the observational equivalent concept also needs to be further extended. 
That is because $\ell_i^m=\ell_i^r$ is not the only solution for equation $h_i^g(\theta_m,\theta_r)=0$; instead, all the solutions of $h_i^g(\theta_m,\theta_r)=0$ would cause the observational equivalent problem. 

Observational equivalent problem is an extreme case and could be easily eliminated by social collaboration. Therefore, if not mentioned particularly, this article will not discuss the observational equivalent case, that means Definition \ref{def:2} and the following definitions and propositions will not cover the special case that $h_i^g(\theta_m,\theta_r)=0$ for some $m\ne r$.

\begin{figure}
	\centering
	\includegraphics[width=0.32\textwidth]{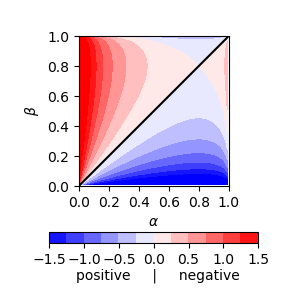}
	\caption{$h_i^g(\ell_i^2, \ell_i^1)$ in $\alpha\beta$ square, $g_i=[0.8, 0.2]$. The two blue parts satisfy the condition in Eq.\eqref{eq:inProbabilityCondition} and belong to the positive signal structure region; the two red parts belong to the negative signal structure region.}
	\label{HG_AB}
\end{figure}

In the binary case, since $\theta_1$ is the real state, $m=2$ is the only case for $m\ne r$. Thus if $h_i^g(\theta_2, \theta_1)<0$, the agent and its signal structure are positive, while if $h_i^g(\theta_2, \theta_1)>0$, the agent and its signal structure are negative. To clarify positive and negative signal structures, $h_i^g(\theta_2, \theta_1)$ in $\alpha\beta$ square when $g_i=[0.8, 0.2]$ is plotted in color map Fig.\ref{HG_AB}.

Since $h_i^g(\ell_i^m, \ell_i^r) = - h_i^g(\ell_i^r, \ell_i^m)$, color map of $h_i^g(\ell_i^m, \ell_i^r)$ is antisymmetric about the axis $\alpha=\beta$. Each of positive region and negative region covers $1/2$ area of the square and consists of a major part and a minor part. The four parts are divided by two observational equivalent lines, one is the straight line $\alpha=\beta$, and the other is the curve line connecting point $(0,1)$ and point $(1,0)$ which passes through the point $(0.8,0.8)$. 
These two lines are perpendicular to each other at point $(0.8,0.8)$, and they are the solution of $h_i^g(\theta_m,\theta_r)=0$, which means signal structures on these two lines may cause the observational equivalent problem. The world signal structure $\alpha=0.8$ passes through the two part of the positive region.

Positive and negative signal structures could also be explained by Fig.\ref{fig:projection}: For a specific signal structure, if $h_i^g(\theta_m,\theta_r)$ has the opposite direction with $h(g_i)$ for all $m\ne r$, the signal structure is positive, or else it is negative. Positive signal structure means negative $h_i^g(\theta_m,\theta_r)$, which also indicates that $h(\ell_i^r)$ and $h(g_i)$ are closer than $h(\ell_i^m)$ and $h(g_i)$.

\subsection{Conservative and Radical Signal Structures}

While Eq.\eqref{eq:inProbabilityCondition} ensures a weak (in probability thus in distribution) convergence, our results show signal structures with stronger conditions could ensure a stronger convergence via learning rule described by Eq.\eqref{eq:isolatedLearningRule}.

\begin{theorem}\label{THM:INDIV_CONSERV}
	Suppose that agent $i$ has positive initial belief on state $\theta_r$, and the private signal structure property $k_i^g(\theta_m, \theta_r)$ defined by
	\begin{equation}\label{eq:kigmr}
		k_i^g(\theta_m, \theta_r)=\log\sum_{s\in S_i} g_{i}(s)\frac{\ell_i^m(s)}{\ell_i^r(s)}
	\end{equation}
	satisfies
	\begin{equation}\label{eq:almostSureCondition}
		k_i^g(\theta_m, \theta_r)<0
	\end{equation}
	for all $m\ne r$, then, agent $i$ learns the state $\theta_r$, \textbf{almost surely}, by the learning rule described in Eq.\eqref{eq:isolatedLearningRule}, which means $\mu_{i,t}\xrightarrow[t\to\infty]{a.s.}\mathbf{1}_r$.

	Proof: \textnormal{See appendix.}
\end{theorem}

Proposition \ref{THM:INDIV_CONSERV} is the almost sure convergence version of Proposition \ref{THM:INDIV_RADICAL}.
Since almost sure convergence indicates convergence in probability, the condition in Eq.\eqref{eq:almostSureCondition} is stronger than the condition in Eq.\eqref{eq:inProbabilityCondition}. More specifically, with Eq.\eqref{eq:higmr}, Eq.\eqref{eq:kigmr}, and convexity of the logarithmic function, by Jenson's inequality, we have
$$
	\log\sum_{s\in S_i} g_{i}(s)\frac{\ell_i^m(s)}{\ell_i^r(s)} > \sum_{s\in S_i}g_i(s)\log\frac{\ell_i^m(s)}{\ell_i^r(s)},
$$
which indicates Eq.\eqref{eq:almostSureCondition} is stronger than Eq.\eqref{eq:inProbabilityCondition}.

We must emphasize that Proposition \ref{THM:INDIV_RADICAL} and Proposition \ref{THM:INDIV_CONSERV} answer Question \ref{ques:individual_learning} in different senses:
Proposition \ref{THM:INDIV_RADICAL} converging in probability and Proposition \ref{THM:INDIV_CONSERV} converging almost surely.
These two types of convergence play an important role in our following definitions and propositions.

By Proposition \ref{THM:INDIV_CONSERV} positive agents could be classified into two types, we name them conservative and radical.
\begin{defn}
	Positive agent and its private signal structure $L_i$ are
	\begin{enumerate}[label=\alph*)]
		\item \textbf{conservative} if $k_i^g(\theta_m, \theta_r)<0$ for all $m\ne r$.
		\item \textbf{radical} if there exists $\hat m(\hat m\ne r)$ such that $k_i^g(\theta_{\hat m},\theta_r)>0$.
	\end{enumerate}
\end{defn}

\begin{figure}
	\centering
	\includegraphics[width=0.32\textwidth]{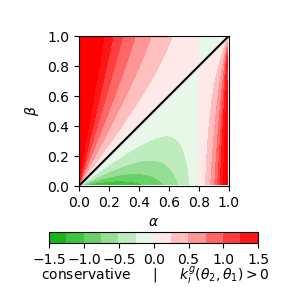}
	\caption{$k_i^g(\ell_i^2, \ell_i^1)$ in $\alpha\beta$ square, $g_i=[0.8, 0.2]$. The two green parts satisfy the condition in Eq.\eqref{eq:almostSureCondition} and belong to the conservative region.}
	\label{fig:KG_AB}
\end{figure}

Fig.\ref{fig:KG_AB} is the color map of $k_i^g(\theta_2, \theta_1)$ in the binary case ($g_i=[0.8, 0.2]$), which shows that the conservative region is between two intersecting lines. One of these two lines is still the line $\alpha=\beta$, part of the boundary of positive and negative regions; the other line is $\alpha=0.8$, the junction between conservative and radical regions. Notice that the perfect signal structures (with $\ell_i^r = g_i^T$) are right at the junction between radical and conservative regions; this means that if an agent has the perfect world signal structure, its signal structure is neither conservative nor radical.

\begin{figure}
	\centering
	\hfill
	\subfloat[$g_i=(0.8,0.2)$.\label{fig:INFORMATIVENESSa}]{%
		\includegraphics[width=0.48\linewidth]{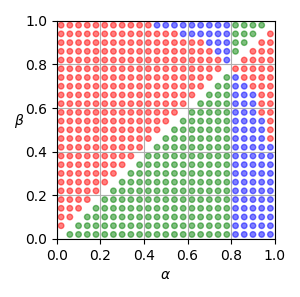}
	}
	\hfill
	\subfloat[$g_i=(0.6,0.4)$.\label{fig:INFORMATIVENESSb}]{%
		\includegraphics[width=0.48\linewidth]{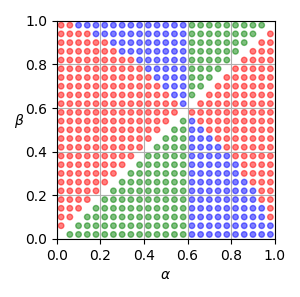}
	}
	\hfill
	\caption{Conservative(green), radical(blue) and negative(red) region in $\alpha\beta$ square with $25\times25$ sample points from $(0.02, 0.02)$ to $(0.98, 0.98)$, with different $g_i$ and informativeness of world signal structure. Especially, Fig.\ref{fig:INFORMATIVENESSa} could be recognized as a qualitative combination of Fig.\ref{HG_AB} and Fig.\ref{fig:KG_AB}.}
	\label{fig:INFORMATIVENESS}
\end{figure}

Fig.\ref{fig:INFORMATIVENESS} marks conservative, radical and negative signal structures in $\alpha\beta$ square with 625 sample points. The positive region is divided into conservative and radical regions by line $\alpha=g_i(s_1)$. By comparing the case when $g_i=[0.8, 0.2]$ (Fig.\ref{fig:INFORMATIVENESSa}, higher $h(g_i)$ means world signals are more informative) with the case when $g_i=[0.6, 0.4]$ (Fig.\ref{fig:INFORMATIVENESSb}, lower $h(g_i)$ means world signals are less informative), it could be noticed that the radical region (blue) is fatter when the world signal structure is less informative, which usually lead to a greater opportunity for an agent's signal structure to locate in the radical region.

Since $\ell_i^r \ne g_i^T$, both conservative and radical agents would not rule out the wrong states as much as the informativeness of the signal accordingly. Specifically, for most of the cases, agents tend to accept the real state when a high signal is observed ($\alpha>\beta$); therefore, the conservative agents always underestimate the informativeness of high signals (the major green part in Fig.\ref{fig:INFORMATIVENESS}), while radical agents always overestimate the informativeness of high signals (the major blue part in Fig.\ref{fig:INFORMATIVENESS}). However, if agents tend to reject the real state when a high signal is observed ($\alpha<\beta$), conservative agents always overestimate the informativeness of high signal (the minor green part in and Fig.\ref{fig:INFORMATIVENESS}), while radical agents always underestimate the informativeness of high signal (the minor blue part in Fig.\ref{fig:INFORMATIVENESS}). 

Semantically, conservative means be more cautious and less arbitrary in decision-making. The reason why we use the word conservative is that the line $\alpha=\beta$ is closer to the conservative region than to the radical region, which means, for conservative agents, knowledge difference between real state and wrong states is less than that of radical agents. It enables more fault tolerance for conservative agents to make decisions.

\section{Social Learning}\label{sec:social_learning}
Now let us switch back to the network context, where the learning rule is represented by Eq.\eqref{eq:learning_rule1}. To guarantee that information can flow from an arbitrary agent in the network to others, we first assume the connectivity of the network in our work.
\begin{asspt}
	\label{asspt:connectivity}
	The underlying social network is strongly connected.
\end{asspt}
This connectivity assumption together with Assumption \ref{asspt:onePositiveInitialBelief} make sure that all agents have strictly positive belief on real state after finite time steps.
The mathematical merit of strong connectivity is that it guarantees that the influence matrix $A$ of the underlying network is irreducible.

\subsection{Social Learning of Conservative Agents}

In a strongly connected network, our results show the social learning version of Proposition \ref{THM:INDIV_CONSERV} holds, which is also the answer to Question \ref{ques:social_learning}.

\begin{theorem}
	\label{THM:NETWORKALMOSTSURE}
	Suppose that all agents' private signal structures are conservative and follow the learning rule described in Eq.\eqref{eq:learning_rule1}, then all agents in the social network learn the real state of the world almost surely, that is, $\mu_{i,t}\xrightarrow[t\to\infty]{a.s.}\mathbf{1}_r$ for all $i$.

	Proof: \textnormal{See appendix.}
\end{theorem}

Proposition \ref{THM:NETWORKALMOSTSURE} gives sufficient conditions for a network to achieve asymptotic learning: Distributed almost sure convergence could lead to the almost sure convergence over a network via non-Bayesian social learning.
Group of conservative individuals could reach consensus and learn the real state collectively without knowing the accurate signal knowledge about the real state. Underestimating or overestimating the informativeness of high or low signal will not ruin the learning of network, for conservative agents could always accumulate correct information over time.

Compared with the asymptotic learning conditions given in \cite{jadbabaie2012non}, the condition that no other state is observationally equivalent to the real state from all agents is no longer needed. The reason for this is that conservative condition in Proposition \ref{THM:NETWORKALMOSTSURE} is different from and slightly stronger than the globally identifiable condition given in \cite{jadbabaie2012non} to adapt our extended model Assumption \ref{asspt:imperfect}. All agents' private signal structures are conservative means that no state is observational equivalent to the real state from the view of each agent. The real state is locally identifiable, which implies that the real state is globally identifiable.

In the binary case, we set up our simulation environment using static undirected ER random network, in which all agents share the same signal space and same world signal structure, which means $g_{i}=g_{j}$ and $S_{i}=S_{j}$ for any $i\ne j$.
In a network of 100 nodes, we set the density of network to 0.1 and ensure its strong connectivity.
Learning rule in Eq.\eqref{eq:learning_rule4} is adopted, which means that all agents share the same self-reliance $\gamma$ and each agent assigns the same weight to all its neighbors when performing information aggregation.
Agents' initial beliefs assigned to $\theta_1$ and $\theta_2$ are uniformly distributed in the interval of $[0,1]$. We assume that $\theta_1$ is the real state, signal $s_1$ appears with the possibility of 80\%, and $s_2$ with 20\%, thus, $g_{i}=[0.8,0.2]$ for all $i$.

\begin{table}
	\caption{Constant signal structures used in simulations.}\label{tb:cs}
	\centering
	\begin{tabular}{*{6}{c}}
		\hline
		& type         & $\alpha$ & $\beta$ & $h_i^g(\theta_m,\theta_r)$ & $k_i^g(\theta_m,\theta_r)$ \\
		\hline
		$L_{(1)}$ & Conservative & 0.6      & 0.4     & $-0.2433$                  & $-0.1823$                  \\
		$L_{(2)}$ & Radical      & 0.9      & 0.1     & $-1.3183$                  & $0.6360$                   \\
		$L_{(3)}$ & Negative     & 0.4      & 0.6     & $0.2433$                   & $0.2877$                   \\
		\hline
	\end{tabular}
\end{table}

To examine the learning process of a network composed of conservative agents, we set all agents' signal structure $L_{(1)}$ in Table.\ref{tb:cs}. The learning process of the network in 500 steps is shown in Fig.\ref{fig:conlearning}.

\begin{figure}
	\begin{center}
		\subfloat[$\gamma=0.1.$\label{fig:conlearning_0.1}]{%
			\includegraphics[width=0.9\linewidth]{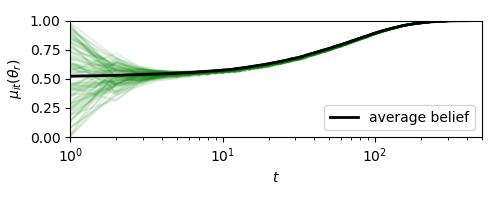}
		}
		\hfill
		\subfloat[$\gamma=0.9.$\label{fig:conlearning_0.9}]{%
			\includegraphics[width=0.9\linewidth]{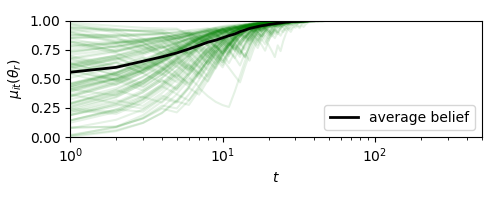}
		}
	\end{center}
	\caption{Belief dynamics of ER random network composed of 100 conservative agents (all agents are with $L_{(1)}$). Each green line represents an agent's belief dynamics. With different self-reliances, conservative agents could reach consensus and convergence at different speeds respectively.}
	\label{fig:conlearning}
\end{figure}

Fig.\ref{fig:conlearning} is the direct demonstration of Proposition \ref{THM:NETWORKALMOSTSURE}. The green lines, which is in accordance with the color of the conservative region in Fig.\ref{fig:INFORMATIVENESSa}, represent agents' belief dynamics on the real state over the time axis. All green lines and average belief are approximating 1, which means the conservative agents could collectively learn the underlying state. In general, while lower self-reliance results in faster consensus (see Fig.\ref{fig:conlearning_0.1}), higher self-reliance results in faster convergence (see Fig.\ref{fig:conlearning_0.9}).

\subsection{Social Learning of Radical Agents}

Proposition \ref{THM:NETWORKALMOSTSURE} is the social version of Proposition \ref{THM:INDIV_CONSERV}; to some extent, its convergence is predictable. Also, it is easy to infer that negative agents always do harm to social learning. In this sub-section, we will focus on the racial agents; a counter-intuitive phenomenon will be shown: Social version of Proposition \ref{THM:INDIV_RADICAL} does not hold, which means network composed of radical agents could not achieve successful learning.

\begin{theorem}
	\label{THM:NETWORKUNCERTAIN}
	Suppose that all agents' private signal structures are radical and follow the learning rule described in Eq.\eqref{eq:learning_rule1}, then all agents remain uncertain almost surely.

	Proof: \textnormal{See appendix.}
\end{theorem} 

Proposition \ref{THM:NETWORKUNCERTAIN} reveals that distributed convergence in probability could not ensure a collective convergence.
If all agents are with radical signal structure $L_{(2)}$ in Table.\ref{tb:cs}, in the same simulation environment, the learning process of 100-nodes network is shown in Fig.\ref{fig:ralearning}. 
It could be found that although the network has no chance to learn a wrong state, the learning result may remain uncertain. 
That is somewhat counter-intuitive because all radical agents could perform success learning independently. This counter-intuitiveness can be viewed as a kind of information cascades \cite{banerjee1992simple, bikhchandani1992theory, rosas2017technological}: Neighbors' misleading observations are overestimated or underestimated by themselves before involving in each agent's subsequent decision-making. 
For each radical agent, it can not completely rule out these misleading information by the accumulation of new observations.

\begin{figure}
	\begin{center}
		\subfloat[$\gamma=0.1.$\label{fig:ralearning_0.1}]{%
			\includegraphics[width=0.9\linewidth]{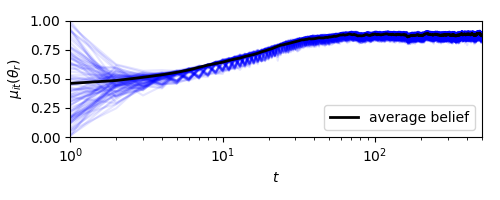}
		}
		\hfill
		\subfloat[$\gamma=0.9.$\label{fig:ralearning_0.9}]{%
			\includegraphics[width=0.9\linewidth]{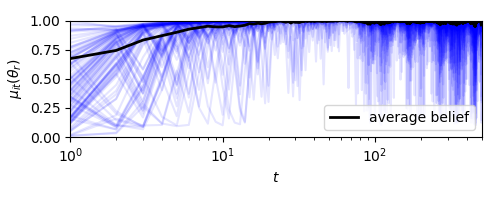}
		}
	\end{center}
	\caption{Belief dynamics of ER random network composed of 100 radical agents (all agents are with $L_{(2)}$). Whenever the common self-reliance is high or low, radical agents could not learn the real state collectively and completely.}
	\label{fig:ralearning}
\end{figure}

Particularly, in Fig.\ref{fig:ralearning}, when learning with lower self-reliance, the agents reach a noisy consensus (see Fig.\ref{fig:ralearning_0.1}). 
In addition, learning with higher common self-reliance is an approximation to the individual learning case; thus agents' average belief on the real state keeps at a high value after about 300 steps (see Fig.\ref{fig:ralearning_0.9}). 
However, higher common self-reliance prevents the network from consensus. In Fig.\ref{fig:ralearning_0.9}, it can be seen that there are always a few numbers of agents in network assign low beliefs on real state after 100 steps.

What's interesting is that in \cite{nedic2017fast} the results show that if each agent uses a log-linear information aggregation rule, then radical agents could learn the real state collectively. That means radical agents could avoid information cascades by taking a more rational information aggregation approach.

\subsection{Learning speed}\label{sec:learning_speed}

In this part, we will discuss the learning speed under the model Assumption \ref{asspt:imperfect}. As a parameter describing the learning process, Learning speed represents how fast agents' beliefs converge to $\mathbf{1}_r$. Learning rate discussed in \cite{jadbabaie2013information} is used to measure the learning speed of a network. Here we introduce the definitions of belief uncertainty and learning rate in \cite{jadbabaie2013information}.

\begin{defn}
	\textbf{Belief uncertainty} of a network at time $t$ is given by
	\begin{equation}\label{eq:beliefUncertainty}
		e_t= \frac{1}{2}\sum_{i=1}^{n}||\mu_{i,t}-\mathbf{1}_r||_1,
	\end{equation}
	and \textbf{learning rate} is defined by
	\begin{equation*}
		\lambda = \liminf_{t\rightarrow\infty}\frac{1}{t}|\log{e_t}|.
	\end{equation*}
\end{defn}

With the above definitions of belief uncertainty and learning rate, our proposition about learning speed is as follows:

\begin{theorem}\label{THM:LEARNINGRATE}
	Suppose all agents in the network are conservative and share the same self-reliance $\gamma$, which means the learning rule is described by Eq.\eqref{eq:learning_rule3}, then the learning rate
	\begin{equation*}
		\lambda\leq\gamma\min_{m\ne r}\sum_{i=1}^{n}v_i |k_i^g(\theta_{m},\theta_r)|
	\end{equation*}
	almost surely.

	Proof: \textnormal{See appendix.}
\end{theorem}

Proposition \ref{THM:LEARNINGRATE} needs to be discussed when the network satisfies the conditions of Proposition \ref{THM:NETWORKALMOSTSURE}; otherwise, the existence of radical and negative agents might ruin social learning. Compared with Proposition 2 in \cite{jadbabaie2013information}, the relative entropy $h_i(g_i,\theta_m)$ is replaced by $|k_i^g(\theta_m,\theta_r)|$ in our Proposition \ref{THM:LEARNINGRATE}.
Since we have
$$
	0 < -k_i^g(\theta_m,\theta_r) < -h_i^g(\theta_m,\theta_r) < h_i(g_i, \theta_m),
$$
conservative agents ensure successful learning by lowering the learning rate.

In the same simulation environment (connected network with 100 nodes and density 0.1), fix $\gamma= 0.5$, for each conservative signal structure point $L^{(\alpha,\beta)}$ in Fig.\ref{fig:INFORMATIVENESSa}, set $L_i = L^{(\alpha,\beta)}$ for all node $i$ in network, after 50 steps' learning, fill the value of $|\log(e_{50})^{(\alpha,\beta)}|$ back into the $\alpha\beta$ square, then the relationship between signal structure and belief uncertainty is shown in Fig.\ref{fig:learning_speed_ab}.

\begin{figure}
	\centering
	\includegraphics[width=0.4\textwidth]{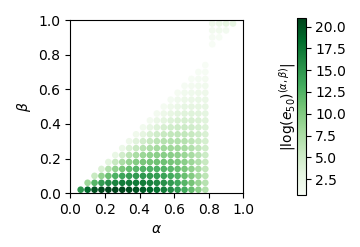}
	\caption{Relationship between conservative signal structure $L^{(\alpha,\beta)}$ and belief uncertainty $|\log(e_{50})^{(\alpha,\beta)}|$. Value of each point $(\alpha,\beta)$ in this colormap represents $|\log(e_{50})^{(\alpha,\beta)}|$ of a network composed of agents whose signal structures are all $L^{(\alpha,\beta)}$.}
	\label{fig:learning_speed_ab}
\end{figure}

\begin{figure}
	\centering
	\includegraphics[width=0.4\textwidth]{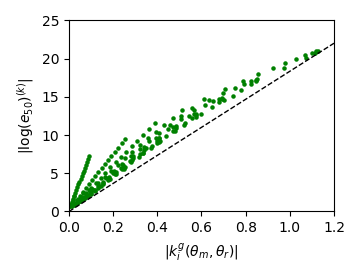}
	\caption{Relationship between $|k_i^g(\theta_2, \theta_1)|$ of conservative signal structure $L^{(\alpha,\beta)}$ and $|\log(e_{50})^{(k)}|$. The x-coordinate of each point is the $|k_i^g(\theta_2, \theta_1)|$ value of each point in Fig.\ref{fig:learning_speed_ab}, and y-coordinate is the corresponding $|\log(e_{50})^{(k)}|$. It appears an approximately linear relationship between $|k_i^g(\theta_2, \theta_1)|$ and $|\log(e_{50})^{(k)}|$.}
	\label{fig:learning_speed_k}
\end{figure}

By comparing Fig.\ref{fig:KG_AB} with Fig.\ref{fig:learning_speed_ab}, it could be noticed that green points in Fig.\ref{fig:learning_speed_ab} and green part of Fig.\ref{fig:KG_AB} are similar, which indicates that there is a close relationship between $|k_i^g(\theta_2, \theta_1)|$ and learning speed.

In addition, Fig.\ref{fig:learning_speed_k} plots $|k_i^g(\theta_2, \theta_1)|$ of each conservative signal structure $L^{(\alpha,\beta)}$ in Fig.\ref{fig:learning_speed_ab} and the corresponding $|\log(e_{50})^{(k)}|$ together. To some extent, it appears a linear relationship between $|k_i^g(\theta_2, \theta_1)|$ and $|\log(e_{50})^{(k)}|$. That is to say, when other conditions such as self-reliance level remain unchanged, the relationship between learning speed of network and level of $|k_i^g(\theta_{\hat m}, \theta_r)|$ among agents approximates to linear, where the $\theta_{\hat m}$ denotes the most indistinguishable alternative state from the view of the whole network. This linear relationship is an indirect demonstration of Proposition \ref{THM:LEARNINGRATE}.

\section{Consideration in actual Social Practices}\label{sec:discuss}

In this section, we first propose some fragilities of non-Bayesian social learning model from a mathematical perspective. 
Then we show the benefits that negative agents can acquire from social learning. 
At last, two adjusting strategies and three decision-making patterns are summarized to explain the collective decision-making mechanism in actual social practices from mathematical and social perspective respectively.

\subsection{Fragilities in Non-Bayesian Social Learning}
Based on mathematical results in Section \ref{sec:inividual_learning} and Section \ref{sec:social_learning}, the fragilities of social learning mechanism are mainly reflected in the three following aspects: 

a). The collective decision-making cannot improve learning speed, for higher social-reliance (lower self-reliance) usually results in lower learning speed. Particularly, for conservative agents and agents with perfect signal structures, social learning is slower than individual learning.

b). Successful individual learning could not ensure successful social learning, for the existence of radical agents might ruin the successful learning over the network. 

c). Agents with perfect signal structure are right at the junction of the conservative region and the radical region in $\alpha\beta$ square (see Fig.\ref{fig:INFORMATIVENESS}), which makes these agents easy to turn radical when endogenous noise is introduced into their private signal structure.

\subsection{Negative agents in Social Learning}

Since fragilities exist in non-Bayesian social learning model, why individuals tend to learn from their neighbors by joining the network in actual social practices? 
One explanation we can find within the frame of non-Bayesian social learning model is that, by calling the individuals together, the network may provide more opportunities for negative agents to learn the real state with the help of their conservative neighbors.

\begin{figure}
	\begin{center}
		\subfloat[$\gamma=0.1.$\label{fig:connglearning_0.1}]{%
			\includegraphics[width=0.9\linewidth]{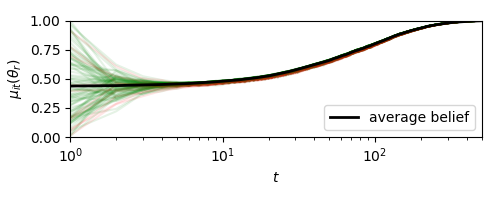}
		}
		\hfill
		\subfloat[$\gamma=0.5.$\label{fig:connglearning_0.5}]{%
			\includegraphics[width=0.9\linewidth]{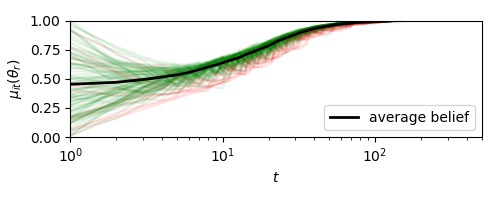}
		}
		\hfill
		\subfloat[$\gamma=0.9.$\label{fig:connglearning_0.9}]{%
			\includegraphics[width=0.9\linewidth]{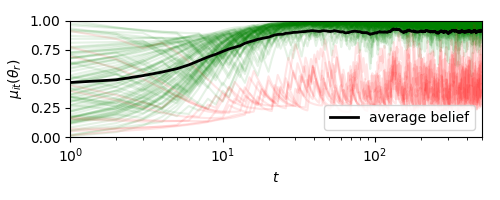}
		}
	\end{center}
	\caption{Belief dynamics of ER random network composed of 90 conservative agents (with $L_{(1)}$, green) and 10 negative agents (with $L_{(3)}$, red) at different self-reliances $\gamma$. Higher $\gamma$ means faster learning (Fig.\ref{fig:connglearning_0.5}) but might result in overconfidence (Fig.\ref{fig:connglearning_0.9}).}
	\label{fig:connglearning}
\end{figure}

Fig.\ref{fig:connglearning} shows the learning process of a network composed of 90 conservative agents (with $L_{(1)}$ in Table \ref{tb:cs}) and 10 radical agents (with $L_{(3)}$ in Table \ref{tb:cs}) in three different self-reliance values $\gamma$. 
When common $\gamma$ shifts from 0.1 to 0.5, the overwhelming conservative agents can help the negative agents to learn the real state (see Fig.\ref{fig:connglearning_0.1},\ref{fig:connglearning_0.5}). Particularly, as the $\gamma$ increases, the learning speed of the network increases.
However, if the $\gamma$ is high enough, the network remains uncertain (see Fig.\ref{fig:connglearning_0.9}). In such a case, high self-reliance lead to overconfidence, and when the overconfidence of negative agents occurs, limited help could be obtained from conservative agents, which further results in a failed learning.

\subsection{Adjustment Strategies}

Negative agents can get benefits from conservative neighbors in social learning; however, insufficient arguments are given to explain why the conservative and radical agents have motivations to join social learning. 

In fact, all three types of agents might exist in the actual social network. 
For each agent, as it does not know the world signal structure, it is also unaware of its private signal structure type. 
However, we find that, although the private signal structure type is unknown to each agent, it would be much easier for the entire network to achieve successful learning if its agents could adjust their private signal structure and self-reliance via some appropriate strategies. These strategies for each agent include:

a). \textit{Be more conservative}. Mathematically, "be more conservative" means each agent should try to reduce the difference between the two signal structure columns which point to the two most possible states (in actual social practice, these two states are very likely to be the real state and the most indistinguishable alternative state), especially when the KL divergences between its belief and neighbors' beliefs are too large. By doing this, both radical and negative agents' signal structures would have smaller $k_i^g (\theta_m, \theta_r)$ on the most indistinguishable alternative state $\theta_m$ (in the binary case, it is closer to $\alpha=\beta$). That is beneficial to both radical and negative agents: Radical agents might have an opportunity to become a conservative agent, and negative agents would be more likely to be persuaded by their conservative neighbors.

b). \textit{Avoid overconfidence}. In the non-Bayesian social learning model, the physical meaning of self-reliance is the self-confidence. Therefore, "avoid overconfidence" means an agent should assign enough weights to their neighbors' beliefs. By doing this, adequate and effective support to a negative agent could be provided by its conservative neighbors. As a result, the collaborative tightness of agents in a social network could improve the learning over the network. This strategy might slow down the learning speed, but it is beneficial to the long run learning result.

\subsection{Social Members' Decision-Making Patterns}

Above two strategies are based on our propositions, simulations, and inferences. To extend the implementation value of these two strategies from the perspective of actual social practices, we conclude three following decision-making patterns of social members:

Pattern A: \textit{Confident and Perfect/Radical}. This pattern matches the radical agents and agents with perfect signal structure. Social members with accurate knowledge (perfect signal structure) and those who dare to overestimate the informativeness of high signals (radical) usually keep strong self-confidence (high self-reliance) in actual social practice. As a form of social communication, they often neglect others' opinions (low social-reliances).

Pattern B: \textit{Cooperative and Conservative}. This pattern matches the conservative agents. Individuals with inaccurate or inadequate knowledge (imperfect signal structure) usually make decisions cautiously (conservatively) and communicate more with their neighbors (high social-reliances).

Pattern C: \textit{Dependent and Negative}. This pattern matches the negative agents. Individuals tend to be suspicious about its knowledge (signal structure) whenever others' opinions have great difference (great KL divergences) to its own. In other words, to negative agents, social pressure might result in weak self-confidence and knowledge suspicion. Therefore, spontaneous adjustments would further be taken to fill gaps between their and their neighbors' beliefs.

Above three decision-making patterns are consistent with our mathematical results and analyses. 
Furthermore, several canonical theoretical and empirical studies in the area of social science, such as self-efficacy \cite{bandura2010self}, social pressure \cite{asch1955opinions}, subjective norms \cite{ajzen1986prediction}, and leadership \cite{couzin2005effective} also provide observation and explanation on these decision-making patterns in the real society. More importantly, these patterns might be the essential components of collective decision-making in human and animal societies, specialized in bringing individual together, ensuring the robustness of the social network, enabling as many individuals as possible to survive each critical moment in both history and future.

\section{Conclusion}
In this paper, under the imperfect private signal structure assumption, agents in the social network are classified into three types: "conservative," "radical," and "negative".
The conservative agents always push a network to the correct decision, while both negative and radical agents are detrimental to the learning of the network. 
A network contains negative and/or radical agents may not be able to learn the real state of the world.
If conservative agents are overwhelming in a network, they could help the rest agents in learning the real state, thus making social learning more robust than individual learning. 
If each member could adjust its signal structure and self-reliance with appropriate strategies, it would be easier for the network to learn the real state synchronically.
In summary, the non-Bayesian social learning model is mathematically fragile and practically effective. The assumption of imperfect private signal structure, along with adjusting strategies and decision-making patterns we summarized could better reflect the learning mechanism of human and animal societies.

This paper studies the learning abilities of different signal structures mathematically. Most of the discussion is based on mathematical propositions and simulations, the details in decision-making patterns need to be verified empirically, and some model errors such as learning rule error are not fully taken into account. 
Thus, model extension researches such as the influence of different learning rules, the performance of different strategies, and empirical researches about decision-making patterns could be included in future works.

\section*{Appendix: Proofs}

First, we give proof of Proposition \ref{THM:INDIV_RADICAL}. Then follows three auxiliary lemmas for a specific agent, which are used to prove Propositions \ref{THM:INDIV_CONSERV} and \ref{THM:NETWORKALMOSTSURE}. After that, proof of Proposition \ref{THM:NETWORKUNCERTAIN} is given by contradiction. The proof of Proposition \ref{THM:LEARNINGRATE} is given in the last sub-section. Proofs of some lemmas and propositions follow similar approaches as those in \cite{jadbabaie2012non} and \cite{jadbabaie2013information}.

\subsection{Proof of Proposition \ref{THM:INDIV_RADICAL}}

\begin{proof}
	By evaluating the ratio of belief on any alternative state to the real state, we have
	\begin{equation*}
		\begin{split}
		\frac{\mu_{i,t}(\theta_m)}{\mu_{i,t}(\theta_r)}&=\frac{\mu_{i,0}(\theta_m)}{\mu_{i,0}(\theta_r)}\prod_{\hat t=1}^{t}\frac{\ell_i^m(\omega_{i,\hat t})}{\ell_i^r(\omega_{i,\hat t})} \\ 
		&=\frac{\mu_{i,0}(\theta_m)}{\mu_{i,0}(\theta_r)}\prod_{n=1}^{N}\left(\frac{\ell_i^m(s_n)}{\ell_i^r(s_n)}\right)^{z_n},
		\end{split}
	\end{equation*}
	in which $z_n$ denote the signal $s_n$ appeared $z_n$ times in $t$ steps. 

	Let random variable sequence $X^m_t=\mu_{i,t}(\theta_m)/\mu_{i,t}(\theta_r)$ and $\epsilon>0$, consider
	\begin{equation*}
		\begin{split}
			P(|X^m_t|\ge\epsilon) &= P\left(\frac{\mu_{i,0}(\theta_m)}{\mu_{i,0}(\theta_r)}\prod_{n=1}^{N}\left(\frac{\ell_i^m(s_n)}{\ell_i^r(s_n)}\right)^{z_n}\ge\epsilon\right) \\
			&= P\left(\sum_{n=1}^{N}z_n\log\frac{\ell_i^m(s_n)}{\ell_i^r(s_n)}\ge\log\frac{\epsilon\mu_{i,0}(\theta_r)}{\mu_{i,0}(\theta_m)}\right),
		\end{split}
	\end{equation*}
	since $z_n\xrightarrow[t\to\infty]{p}tg_i(s_n)$,
	\begin{equation*}
		\begin{split}
			&\lim_{t\rightarrow\infty} P(|X^m_t|\ge\epsilon) \\
			&= \lim_{t\rightarrow\infty} P\left(\sum_{n=1}^{N}tg_i(s_n)\log\frac{\ell_i^m(s_n)}{\ell_i^r(s_n)}\ge\log\frac{\epsilon\mu_{i,0}(\theta_r)}{\mu_{i,0}(\theta_m)}\right),
		\end{split}
	\end{equation*}
	then with Eq.\eqref{eq:inProbabilityCondition}, we have $\lim_{t\rightarrow\infty}P(|X^m_t|\ge\epsilon)=0$ for any $m\ne r$, which indicates that $\mu_{i,t}\xrightarrow[t\to\infty]{p}\mathbf{1}_r$.
\end{proof}

\subsection{Auxiliary lemmas for specific agent}

\begin{lemma}
	For any agent $i$'s private signal structure satisfies Eq.\eqref{eq:almostSureCondition}, we have $\mathbb{E}\left[ \left. \ell_i^r(\omega_{i,t+1})/d_{i,t} (\omega_{i,t+1}) \right|\mathcal{F}_t \right] > 1$ for arbitrary prior belief $ \mu_{i,t} $.
	\label{lemma:1}
\end{lemma}

\begin{proof}
	The condition in Eq.\eqref{eq:almostSureCondition}, which is $$\sum_{s\in S_i} g_{i}(s)\frac{\ell_i^m(s)}{\ell_i^r(s)}<1$$ for $m\ne r$, means for arbitrary $\mu_{i,t}$ we have
	\begin{equation*}
		g_{i} \big[diag^{-1}[\ell_i^r]\ell_i^1, diag^{-1}[\ell_i^r]\ell_i^2, \cdots \big] \mu_{i,t} < 1,
	\end{equation*}
	where the left side
	\begin{equation*}
		\begin{split}
			& g_{i} \big[diag^{-1}[\ell_i^r]\ell_i^1, diag^{-1}[\ell_i^r]\ell_i^2, \cdots \big] \mu_{i,t} \\
			=& \sum_{s\in S_i} {g_{i}(s)\frac{d_{i,t} (s)}{\ell_i^r(s)}} \\
			=& \mathbb{E}\left[ \left. \frac{d_{i,t} (\omega_{t+1})}{\ell_i^r(\omega_{t+1})} \right|\mathcal{F}_t \right].
		\end{split}
	\end{equation*}

	Since $f(x)=1/x$ is a convex function, Jensen's inequality implies that
	$$
		\mathbb{E}\left[ \left. \frac{\ell_i^r(\omega_{i,t+1})}{d_{i,t} (\omega_{i,t+1})} \right|\mathcal{F}_t \right]\ge {{\left( \mathbb{E}\left[ \left. \frac{d_{i,t} (\omega_{i,t+1})}{\ell_i^r(\omega_{i,t+1})} \right|\mathcal{F}_t \right] \right)}^{-1}},
	$$
	which leads to
	$$
		\mathbb{E}\left[ \left. \frac{\ell_i^r(\omega_{i,t+1})}{d_{i,t} (\omega_{i,t+1})} \right|\mathcal{F}_t \right]\ge 1.
	$$
\end{proof}

\begin{lemma}{If we have $\mathbb{E}\left[ \left. \ell_i^r(\omega_{i,t+1})/d_{i,t} (\omega_{i,t+1}) \right|\mathcal{F}_t \right]\to 1$ for agent $i$, then we have $d_{i,t}(s)\xrightarrow[t\to\infty]{a.s.}\ell_i^r(s)$ for all possible signal $s\in S_i$.}
	\label{lemma:2}
\end{lemma}

\begin{proof} Note that both $\ell_i^r$ and $d_{i,t}$ are probability measures on $S_{i}$, and we have
	$$
		\sum_{s\in S_i} \ell_i^r(s) = \sum_{s\in S_i} d_{i,t}(s)= 1.
	$$
	Consider
	\begin{equation*}
		\begin{split}
			& \mathbb{E}\left[ \left. \frac{\ell_i^r(\omega_{i,t+1})}{d_{i,t} (\omega_{,it+1})} \right|\mathcal{F}_t \right]-1 \\
			=& \sum_{s\in S_i} {g_{i}(s)\frac{ \ell_i^m(s)-d_{i,t} (s) }{d_{i,t} (s)}} + \sum_{s\in S_i} {\ell_i^r(s)}-\sum_{s\in S_i} {d_{i,t} (s)} \\
			=& \sum_{s\in S_i} {\frac{[ g_{i}(s) + d_{i,t}(s)][ \ell_i^r(s)-d_{i,t}(s)]}{d_{i,t} (s)}} \xrightarrow[t\to\infty]{a.s.}0.
		\end{split}
	\end{equation*}
	Now we have
	$$
		\sum_{s \in S_i} {g_i(s)\frac{\ell_i^r(s) - d_{i,t}(s)}{d_{i,t} (s)}}\xrightarrow[t\to\infty]{a.s.}0,
	$$
	and at the same time we have
	$$
		\sum_{s \in S_i} {\frac{[ g_i(s) + d_{i,t} (s) ][ \ell_i^s(s) - d_{i,t} (s) ]}{d_{i,t} (s)}}\xrightarrow[t\to\infty]{a.s.}0,
	$$
	that means for arbitrary $c$,
	\begin{equation}\label{eq:equations}
		\sum_{s\in S_i} {\frac{[ c g_i(s) + d_{i,t} (s) ][ \ell_i^r(s) - d_{i,t} (s) ]}{d_{i,t} (s)}}\xrightarrow[t\to\infty]{a.s.}0.
	\end{equation}
	Because $d_{i,t} (s)$ is positive for all $s\in S_{i}$, Equation set Eq.\eqref{eq:equations} only has zero solution, thus we have $d_{i,t} (s)\xrightarrow[t\to\infty]{a.s.}\ell_i^r(s)$ for all possible signals $s\in S_{i}$.
\end{proof}

\begin{lemma}
	For any agent $i$'s private signal structure $L_{i}$ satisfying Eq.\eqref{eq:almostSureCondition}, if we have $d_{i,t} (s)\xrightarrow[t\to\infty]{a.s.}\ell_i^r(s)$ for all possible signal $s\in S_{i}$, then $\mu_{i,t} \xrightarrow[t\to\infty]{a.s.}\mathbf{1}_r$.
	\label{lemma:3}
\end{lemma}

\begin{proof}
	The Bayesian part of the update is distribution over set $\Theta$ means
	\begin{equation*}
		\begin{split}
			&\sum_{m=1}^{M}{ \mu_{i,t} ({\theta_m})\frac{\ell_i^m(s)}{d_{i,t} (s)}} \\
			=& \sum_{m\ne r}{ \mu_{i,t} ({\theta_m})\frac{\ell_i^m(s)}{d_{i,t} (s)}} + \mu_{i,t} (\theta_r)\frac{\ell_i^r(s)}{d_{i,t} (s)} \\
			=& 1.
		\end{split}
	\end{equation*}
	With condition $d_{i,t} (s)\xrightarrow[t\to\infty]{a.s.}\ell_i^r(s)$ for all $s \in S_i$, we have
	$$
		\sum_{m\ne r}{ \mu_{i,t} ({\theta_m})\frac{\ell_i^m(s)}{d_{i,t} (s)}} + \mu_{i,t} (\theta_r)\xrightarrow[t\to\infty]{a.s.}1,
	$$
	leading to
	$$
		\sum_{m\ne r}{ \mu_{i,t} ({\theta_m})\frac{\ell_i^m(s)}{d_{i,t} (s)}}-\left( 1- \mu_{i,t} (\theta_r) \right)\xrightarrow[t\to\infty]{a.s.}0.
	$$
	Because $\mu_{i,t}\in\Delta\Theta$, we have $\mu_{i,t} (\theta_r)=1-\sum_{m\ne r}{ \mu_{i,t} ({\theta_m})}$, thus
	\begin{equation}
		\sum_{m\ne r}{ \mu_{i,t} ({\theta_m})\left( 1-\frac{\ell_i^m(s)}{\ell_i^r(s)} \right)}\xrightarrow[t\to\infty]{a.s.}0,
		\label{eq:beliefEquation}
	\end{equation}
	for $s \in S_{i}$. Note that in Eq.\eqref{eq:almostSureCondition}, $g_{i}\in\Delta S_i$, means $\sum_{s\in S_i} {g_i(s)}=1$, which leads to
	\begin{equation}
		\sum_{s\in S_i} {g_i(s)\left( 1-\frac{\ell_i^m(s)}{\ell_i^r(s)} \right)}>0,
		\label{eq:globalEquation}.
	\end{equation}
	for $m=1,2,\cdots M\text{, }m\ne r$.

	Now, let $\mathcal{L}_i$ be a $|S_i|\times|\Theta|$ matrix of which the entry at the $x$-th row and $y$-th column is $1-\ell_i^y(s_x)/\ell_i^r(s_x)$, and entries of $r$-th column are zero, then Eq.\eqref{eq:beliefEquation} and Eq.\eqref{eq:globalEquation} could be merged into
	\begin{equation}
		\mathcal{L}_i \mu_{i,t} \xrightarrow[t\to\infty]{a.s.}{\mathbf{0}}
		\label{eq:translatedEquation}
	\end{equation}
	and
	$$
		{[{g_{i}}\mathcal{L}_i]}_m>\text{0},\text{ for all }m\ne r.
	$$
	Let an $m-$entry error vector $\varepsilon$ satisfies
	$$\varepsilon_m =
		\begin{cases}
			[{g_{i}}\mathcal{L}_i]_m & m\ne r \\
			0                        & m = r  \\
		\end{cases},
	$$
	multiplying both sides of Eq.\eqref{eq:translatedEquation} by $g_{i}$ from left leads to $\varepsilon \mu_{i,t} \xrightarrow[t\to\infty]{a.s.}0,$ which means $ \mu_{i,t} ({\theta_m})\xrightarrow[t\to\infty]{a.s.}0$ for all $m\ne r$, then we have $ \mu_{i,t} \xrightarrow[t\to\infty]{a.s.}\mathbf{1}_r$.
\end{proof}

\subsection{Proof of Proposition \ref{THM:INDIV_CONSERV}}

\begin{proof}
	For a specific agent $i$, evaluate Eq.\eqref{eq:isolatedLearningRule} at the real state $\theta_r$,
	$$
		{{\mu }_{i,t+1}}(\theta_r)= \mu_{i,t} (\theta_r)\frac{\ell_i^r(\omega_{i,t+1})}{d_{i,t} (\omega_{i,t+1})},
	$$
	thus
	$$
		\mathbb{E}\left[ \left. {{\mu }_{i,t+1}}(\theta_r) \right|\mathcal{F}_t \right]= \mu_{i,t} (\theta_r)\mathbb{E}\left[ \left. \frac{\ell_i^r(\omega_{i,t+1})}{d_{i,t} (\omega_{i,t+1})} \right|\mathcal{F}_t \right].
	$$
	Lemma \ref{lemma:1} guarantees
	\begin{equation*}
		\begin{split}
			&\mathbb{E}\left[ \left. {{\mu }_{i,t+1}}(\theta_r) \right|\mathcal{F}_t \right]- \mu_{i,t} (\theta_r) \\
			&= \mu_{i,t} (\theta_r)\left( \mathbb{E}\left[ \left. \frac{\ell_i(^r\omega_{i,t+1}|}{d_{i,t} (\omega_{i,t+1})} \right|\mathcal{F}_t \right]-1 \right) \ge 0,
		\end{split}
	\end{equation*}
	thus ${{\mu }_{i,t+1}}(\theta_r)$ is a submartingale, and it converges almost surely. Due to the fact that $\mu_{i,t} (\theta_r)$ is non-negative, by the dominated convergence theorem for conditional expectations, we have
	$$
		\mathbb{E}\left[ \left. \frac{\ell_i^r(\omega_{i,t+1})}{d_{i,t} (\omega_{i,t+1})} \right|\mathcal{F}_t \right]\xrightarrow[t\to\infty]{a.s.}1,
	$$
	then Lemma \ref{lemma:2} and Lemma \ref{lemma:3} guarantee $ \mu_{i,t} \xrightarrow[t\to\infty]{a.s.}\mathbf{1}_r$.
\end{proof}

\subsection{Proof of Proposition \ref{THM:NETWORKALMOSTSURE}}

We first introduce another lemma.
\begin{lemma}\label{lemma:4}
	Suppose $v$ is any non-negative left eigenvector of the influence matrix $A$ corresponding to its unit eigenvalue, then $v{{\mu }_{t+1}}(\theta_r)$ is a submartingale, it converges almost surely.
\end{lemma}

\begin{proof}
	Evaluate Eq.\eqref{eq:learning_rule2} at the real state and multiply both sides by $v$ from the left, we have
	\begin{equation*}
		\begin{split}
			&v{{\mu }_{t+1}}(\theta_r)=v\mu_t(\theta_r) \\
			&+\sum_{i=1}^N{v_i a_{ii}\mu_{i,t}(\theta_r)\left[ \frac{\ell_i^r(\omega_{i,t+1})}{d_{i,t} (\omega_{i,t+1})}-1 \right]}.
		\end{split}
	\end{equation*}
	thus
	\begin{equation}
		\begin{split}
			&\mathbb{E}\left[ \left. v{{\mu }_{t+1}}(\theta_r) \right|\mathcal{F}_t \right] =v\mu_t(\theta_r)\\
			&+ \sum_{i=1}^N{v_i a_{ii}\mu_{i,t}(\theta_r)\mathbb{E}\left[ \left. \frac{\ell_i^r(\omega_{i,t+1})}{d_{i,t} (\omega_{i,t+1})}-1 \right|\mathcal{F}_t \right]}.
		\end{split}
		\label{eq:dyanmicAsMartingale}
	\end{equation}
	Lemma \ref{lemma:1} guarantees
	$$\mathbb{E}\left[ \left. \frac{\ell_i^r(\omega_{i,t+1})}{d_{i,t} (\omega_{i,t+1})}-1 \right|\mathcal{F}_t \right]\ge 0.$$
	Strongly connected topology ensures the influence matrix is irreducible; thus, $v_i>0$ and ${{\mu }_{i,m}}(t)>0$ for all $i\in \mathcal{N}$ and all ${\theta_m}\in \Theta$, the second part on the right of Eq.\eqref{eq:dyanmicAsMartingale} is positive which indicate $v{{\mu }_{t}{(\theta_r)}}$ is a submartingale. Hence, it converges almost surely.
\end{proof}

\begin{proof}[Proof of Proposition \ref{THM:NETWORKALMOSTSURE}]
	Consider Eq.\eqref{eq:learning_rule2} together with the conclusion from Lemma \ref{lemma:4}, we have
	$$
		\sum_{i=1}^N{v_i a_{ii} \mu_{i,t} (\theta_r)\left[ \frac{\ell_i^r(\omega_{i,t+1})}{d_{i,t} (\omega_{i,t+1})}-1 \right]}\xrightarrow[t\to\infty]{a.s.}0.
	$$
	Therefore by the dominated convergence theorem for conditional expectations,
	$$
		\sum_{i=1}^N{v_i a_{ii} \mu_{i,t} (\theta_r)\left( \mathbb{E}\left[ \left. \frac{\ell_i^r(\omega_{i,t+1})}{d_{i,t} (\omega_{i,t+1})} \right|\mathcal{F}_t \right]-1 \right)}\xrightarrow[t\to\infty]{a.s.}0.
	$$
	Since $ a_{ii}>0$, entries of ${{v}^{T}}$ and $ \mu_{i,t} ({\theta_{k}})$ are all strictly positive, we have
	$$
		\mathbb{E}\left[ \left. \frac{\ell_i(\omega_{i,t+1}|{\theta_{k}})}{d_{i,t} (\omega_{,it+1})} \right|\mathcal{F}_t \right]\xrightarrow[t\to\infty]{a.s.}1,
	$$
	then, Lemma \ref{lemma:2} and Lemma \ref{lemma:3} guarantee $ \mu_{i,t} \xrightarrow[t\to\infty]{a.s.}\mathbf{1}_r$ for all $i\in \mathcal{N}$.
\end{proof}

\subsection{Proof of Proposition \ref{THM:NETWORKUNCERTAIN}}

\begin{proof}
	We prove Proposition \ref{THM:NETWORKUNCERTAIN} by contradiction. If the network could learn the real state, it must achieve consensus and convergence. Therefore, we need to prove that when the network reaches consensus, the network could not converge to $\mathbf{1}_r$.

	If consensus occurs, members in the network would share the same belief profile. Examine learning rule Eq.\eqref{eq:learning_rule1} at the consensus status, Eq.\eqref{eq:learning_rule1} could be rewritten as:
	$$
		\mu_{i,t+1}(\theta_m) = a_{ii}\mu_{i,t}(\theta_m)\frac{{\ell_i^m}(\omega_{i,t+1})}{d_{i,t}(\omega_{i,t+1})} + (1-a_{ii})\mu_{i,t}(\theta_m)
	$$
	for all $m$. 
	Then we have
	\begin{equation}\label{eq:proof_5_1}
		\begin{split}
		&\mathbb{E}{[\mu_{i,t+1}(\theta_r)|\mathcal{F}_t]}-\mu_{i,t}(\theta_r) \\
		&= a_{ii}\mu_{i,t}(\theta_r)\left(\sum_sg_i(s) \frac{{\ell_i^r}(s)}{d_{i,t}(s)} - 1 \right) \\
    	&=a_{ii}\mu_{i,t}(\theta_r)\left(\sum_sg_i(s) \frac{{\ell_i^m}(s)}{\sum\ell_i^m(s)\mu_{i,t}(\theta_m)} - 1\right).
		\end{split}
	\end{equation}
	Let $\mu_{i,t}(\theta_r)=1-\epsilon$, $\mu_{i,t}(\theta_{\hat m})=\epsilon$, and $\mu_{i,t}(\theta_m)=0$ for all $m\ne r,\hat m$, then the items in brackets on the right side of \eqref{eq:proof_5_1} is a function of $\epsilon$ as following:
	$$
		f(\epsilon) = \sum_sg_i(s) \frac{1}{1-\epsilon + \epsilon\frac{\ell_i^{\hat m}(s)}{{\ell_i^r}(s)}} - 1,
	$$
	moreover, its directive is
	$$
		f^\prime(\epsilon) = \sum_sg_i(s) \frac{1-\frac{\ell_i^{\hat m}(s)}{{\ell_i^{\hat m}}(s)}}{\left(1-\epsilon + \epsilon\frac{\ell_i^{\hat m}(s)}{{\ell_i^r}(s)}\right)^2}.
	$$
	Evaluate $f^\prime(\epsilon)$ at $\epsilon=0$, we have
	$$
	f^\prime(0)=\sum_sg_i(s)\left( 1-\frac{\ell_i^{\hat m}(s)}{{\ell_i^r}(s)}\right) = 1 - \sum_s g_i(s)\frac{\ell_i^{\hat m}(s)}{{\ell_i^r}(s)}.
	$$
	Since agent is radical, we have $f^\prime(0)<0$, then together with the fact $f(0)=0$ and continuity of $f(\epsilon)$, there exists an interval $(0, c),c>0$ in which $f(\epsilon)<0$. Thus, $\mu_{i,t}(\theta_r)$ is a supermartingale in the corresponding interval, which contradict with $\mu_{i,t}(\theta_r)\xrightarrow{t\rightarrow\infty}1$.
\end{proof}

\subsection{Proof of Proposition \ref{THM:LEARNINGRATE}}

\begin{proof}
	For any $m\ne r$, taking logarithms of both sides of Eq.\eqref{eq:learning_rule3} and using Jensen's inequality implies
	$$
		\log\mu_{i,t+1}(\theta_m)\geq \gamma \log\bigg(\frac{\ell_i^m(\omega_{i,t+1})}{d_{i,t}(\omega_{i,t+1})}\bigg) + \sum_{j=1}^n a_{ij}\log \mu_{j,t}(\theta_m)
	$$
	Since $vA=v$, let $x_{t}(\theta_m) = \sum_{i=1}^nv_i\log\mu_{i,t}(\theta_m)$, multiplying both sides of the above inequality by $v_i$ and summing up over $i$ lead to
	\begin{equation}
		\label{equation:beforeAddingUp}
		x_{t+1}(\theta_m) - x_{t}(\theta_m) \ge q_t(\theta_m),
	\end{equation}
	where
	$$
		q_t(\theta_m)=\gamma\sum_{i=1}^{n}v_i\log\bigg(\frac{\ell_i^m(\omega_{i,t+1})}{d_{i,t}(\omega_{i,t+1})}\bigg).
	$$
	Then by summing up over $t$, Eq.\eqref{equation:beforeAddingUp} leads to
	\begin{equation}
		\label{equation:beforeTakingLimit}
		\frac{1}{T}(x_{T}(\theta_m)-x_{0}(\theta_m)) \ge \frac{1}{T} \sum_{t=0}^{T-1}q_t(\theta_m).
	\end{equation}
	Lemma \ref{lemma:2} indicates $d_{i,t}(s)\xrightarrow[t\to\infty]{a.s.}\ell_i^r(s)$ for all $s$, hence
	$$
		p_t(\theta_m)=\gamma\sum_{i=1}^{n}v_i\log\bigg(\frac{\ell_i^m(\omega_{i,t+1})}{\ell_i^r(\omega_{i,t+1})}\bigg)\xrightarrow[t\to\infty]{a.s.}q_t(\theta_m).
	$$
	Since $\lim_{t\rightarrow\infty} \frac{1}{t} x_{0}(\theta_m)=0$, taking the limit of both sides of Eq.\eqref{equation:beforeTakingLimit} as $t\rightarrow\infty$ implies
	\begin{equation*}
		\begin{split}
			\limsup_{T\rightarrow\infty} \frac{1}{T} x_{T}(\theta_m) \ge& \lim_{T\rightarrow\infty} \frac{1}{T} \sum_{t=0}^{T-1}(q_t(\theta_m) -p_t(\theta_m)) \\
			+& \lim_{T\rightarrow\infty}\sum_{t=0}^{T-1}p_t(\theta_m).
		\end{split}
	\end{equation*}
	Given that $q_t(\theta_m)-p_t(\theta_m)$ converges to zero almost surely, the first term on the right-hand side of the above inequality is equal to zero almost surely.
	Furthermore, by the strong law of large numbers, the second term on the right-hand side is equal to $\mathbb{E}[p_t(\theta_m)]$. Therefore, we have
	\begin{equation}\label{eq:firstEquality}
		\begin{split}
			\limsup_{T\rightarrow\infty}\frac{1}{T} &\sum_{i=1}^nv_i\log\mu_{i,t} (\theta_m) \ge\mathbb{E}[p_t(\theta_m)] \\
			=& \gamma\sum_{i=1}^{n}v_ih_i^g(\theta_m, \theta_r).
		\end{split}
	\end{equation}
	By belief uncertainty definition Eq(\ref{eq:beliefUncertainty}), we have
	$$
		e_t=\sum_{i=1}^{n}\sum_{m\ne r}\mu_{i,t}(\theta_m),
	$$
	and as a result
	\begin{equation*}
		\begin{split}
			\log e_t\geq& \log\big(\max_{m\ne r}\max_i\mu_{i,t}(\theta_m)\big) \\
			=&\max_{m\ne r}\max_i\log\mu_{i,t}(\theta_m) \\
			\geq&\max_{m\ne r}\sum_{i=1}^n v_i\log\mu_{i,t}(\theta_m).
		\end{split}
	\end{equation*}
	Thus with Eq(\ref{eq:firstEquality}), for all $m\ne r$,
	\begin{equation*}
		\begin{split}
			\limsup_{t\rightarrow\infty}\frac{1}{t}\log e_t\geq& \limsup_{t\rightarrow\infty}\frac{1}{t}\max_{m\ne r}\sum_{i=1}^n v_i\log\mu_{i,t}(\theta_m) \\
			\geq& \gamma \max_{m\ne r}\sum_{i=1}^{n}v_i h_i^g(\theta_m,\theta_r) \\
			\geq& \gamma \max_{m\ne r}\sum_{i=1}^{n}v_i k_i^g(\theta_m,\theta_r).
		\end{split}
	\end{equation*}
	almost surely, where the second inequality is a consequence of Eq.(\ref{eq:firstEquality}). Consequently,
	\begin{equation*}
		\lambda=\liminf_{t\rightarrow\infty}\frac{1}{t}|\log e_t|\leq \gamma \min_{m\ne r}\sum_{i=1}^{n}v_i|k_i^g(\theta_m,\theta_r)|
	\end{equation*}
	almost surely.
\end{proof}

\bibliography{arxiv.bib}
\bibliographystyle{IEEEtran}

\end{document}